\documentclass[3p]{elsarticle}
\usepackage{amsmath}
\usepackage{amsthm}
\usepackage{amssymb}
\usepackage{amsfonts}
\usepackage{comment}
\usepackage{color}
\usepackage{tabularx}
\usepackage{graphicx}
\usepackage[latin2]{inputenc}

\newcommand{\executeiffilenewer}[3]{%
\ifnum\pdfstrcmp{\pdffilemoddate{#1}}%
{\pdffilemoddate{#2}}>0%
{\immediate\write18{#3}}\fi%
} 

\newcommand{\svg}[2]{\def\svgwidth{#1}%
\executeiffilenewer{#2.svg}{#2.pdf}%
{inkscape -z -D --file=#2.svg %
--export-pdf=#2.pdf --export-latex}%
{\input{#2.pdf_tex}}}

\newtheorem{theorem}{Theorem}
\newtheorem{lemma}[theorem]{Lemma}

\newtheorem{definition}[theorem]{Definition}
\newtheorem{observation}[theorem]{Observation}
\newtheorem{proposition}[theorem]{Proposition}

\DeclareMathAlphabet{\mathcal}{OMS}{cmsy}{m}{n}

\renewcommand{\S}{\mathcal{S}}

\newcommand{\X}{\mathcal{X}}
\newcommand{\F}{\mathcal{F}}

\newcommand{\nonedge}{{\textup{\textsf{\small nonedge}}}}
\newcommand{\nondeg}{{\textup{\textsf{\small nondeg}}}}
\newcommand{\size}{\textup{\textsf{\small size}}}

\newcommand{\problemq}[3]{
\begin{center}
\fbox{\parbox{0.7\linewidth}{
#1

\noindent\begin{tabularx}{\linewidth}{rX}
Input: & #2\\
Find: & #3
\end{tabularx}}}
\end{center}
}

\newcommand{\problemqz}[3]{
\begin{center}
\fbox{\parbox{0.92\linewidth}{
#1

\noindent\begin{tabularx}{\linewidth}{rX}
Input: & #2\\
Find: & #3
\end{tabularx}}}
\end{center}
}

\newcommand{\partition}[3]{$(#1,#2,#3)$-partition}
\newcommand{\partmupq}{\partition{\mu}{p}{q}}
\newcommand{\partitionsc}[3]{$(#1,#2,#3)$-\textsc{Partition}}
\newcommand{\partmupqsc}{\partitionsc{\mu}{p}{q}}

\begin{document}
\begin{frontmatter}
\title{Clustering with Local Restrictions}
\author[dl]{Daniel Lokshtanov}
\ead{daniello@ii.uib.no}
\author[dm]{D\'aniel Marx\fnref{fn2}}
\ead{dmarx@cs.bme.hu}
\address[dl]{Department of Informatics, University of Bergen, Bergen, Norway.}
\address[dm]{Computer and Automation Research Institute, Hungarian Academy of
Sciences (MTA SZTAKI), Budapest, Hungary.}
\fntext[fn2]{Research supported by ERC Starting Grant PARAMTIGHT.}

\begin{abstract}
  We study a family of graph clustering problems where each cluster has to satisfy a certain local requirement. Formally, let $\mu$ be a
  function on the subsets of vertices of a graph $G$. In the
  \partmupqsc\ problem, the task is to find a partition of the
  vertices into clusters where each cluster $C$ satisfies the requirements that (1)
  at most $q$ edges leave $C$ and (2) $\mu(C)\le p$. Our first result
  shows that if $\mu$ is an {\em arbitrary} polynomial-time computable
  monotone function, then \partmupqsc\ can be solved in time
  $n^{O(q)}$, i.e., it is polynomial-time solvable
  {\em for every fixed $q$}. We study in detail three concrete functions $\mu$
  (the number of vertices in the cluster, number of nonedges in the cluster, maximum number of non-neighbors a vertex has in the cluster), which correspond to
   natural clustering problems. For these functions, we show
  that \partmupqsc\ can be solved in time $2^{O(p)}\cdot n^{O(1)}$ and
  in time $2^{O(q)}\cdot n^{O(1)}$ on $n$-vertex graphs, i.e., the problem is
  fixed-parameter tractable parameterized by $p$ or by $q$.
\end{abstract}

\end{frontmatter}
\section{Introduction}

Partitioning objects into clusters or similarity classes is an important task in various
applications such as data mining, facility location, interpreting
experimental data, VLSI design, parallel computing, and many more. The partition has to
satisfy certain constraints: typically, we want to ensure that objects
in a cluster are ``close'' or ``similar'' to each other and/or objects
in different clusters are ``far'' or ``dissimilar.'' Additionally, we
may want to partition the data into a certain prescribed number $k$ of clusters,
or we may have upper/lower bounds on the size of the
clusters. Different objectives and different distance/similarity
measures give rise to specific combinatorial problems.

Correlation clustering
\cite{AilonCN08,DBLP:journals/ml/BansalBC04,DBLP:conf/stacs/MathieuSS10,mathieu-Schudy-correlation}
deals with a specific form of similarity measure: for each pair of
objects, we know that either they are similar or dissimilar. This
means that the similarity information can be expressed as an
undirected graph, where the vertices represent the objects and similar
objects are adjacent.  In the ideal situation every connected
component of the graph is a clique, in which case the components form
a clustering that completely agrees with the similarity information.
However, due to inconsistencies in the data or experimental errors,
such a perfect partitioning might not always be possible.  The goal in
correlation clustering is to partition the vertices into an arbitrary
number of clusters in a way that agrees with the similarity
information as much as possible: we want to minimize the number of
pairs for which the clustering disagrees with the input data (i.e.,
similar pairs that are put into different clusters, or dissimilar
pairs that are clustered together).

In many cases, such as in variants of the correlation clustering problem defined in
the previous paragraph, the objective is to minimize the total error of the
solution. Thus the goal is to find a solution that is good in a global sense, but
this does not rule out the possibility that the solution contains clusters that
are very bad. In this paper, the opposite approach is taken: we want
to find a partition where each cluster is ``good'' in a certain
local sense. This means that the partition has to satisfy a set of local
constraints on each cluster, but we do not try to optimize the total
fitness of clusters.

The setting in this paper is the following. We want to partition the
input $n$-vertex, $m$-edge graph into an arbitrary number of clusters such that (1) at most $q$
edges leave each cluster, and (2) each cluster induces a graph that
is ``cluster-like.'' Defining what we mean by the abstract notion of
cluster-like gives rise to a family of concrete problems.  Formally,
let $\mu$ be a function that assigns a nonnegative integer to each
subset of vertices in the graph and let us require $\mu(X)\le p$ for
every cluster $X$ of the partition. There are many reasonable choices
for the measure $\mu$ that correspond to natural problems. In
particular, in this paper we will obtain concrete results for the
following three measures:
\begin{itemize}
\item $\nonedge(X)$ is the number of nonedges induced by $X$,
\item $\nondeg(X)$ is the maximum degree of the {\em complement} of
  the graph induced by $X$ (i.e., each vertex of $X$ is adjacent to
  all but at most $\nondeg(X)$ other vertices in $X$), and
\item $\size(X)=|X|$ is the number of vertices of $X$.
\end{itemize}
The first two functions express that each cluster should induce a graph that is close to being a clique. Specifically, a vertex set $S$ such that $\nonedge(S) \leq p$ is called a $p$-defective clique, while a vertex set $S$ such that $\nondeg(S) \leq p$ is called a $p$-plex. These generalizations of cliques have been studied in the context of clustering~\cite{GuoKKU11,GuoKNU10,DBLP:journals/algorithmica/BevernMN12}, as well as in other contexts~\cite{BalasundaramBH11}. The third function only requires that each cluster is small. While this is not really a natural requirement for a clustering problem, partitioning graphs into small vertex sets such that each has few outgoing edges has applications in Field Programmable Data Array design~\cite{DBLP:journals/dm/LangstonP98} and hence is of independent interest. For a given function $\mu$ and integers $p$ and $q$, we
denote by \partmupqsc\ the problem of partitioning the
vertices into clusters such that at most $q$ edges leave each cluster
and $\mu(X)\le p$ for every cluster.
Note that by solving this
problem, we can also solve the optimization version where the goal is
to minimize $q$ subject to a fixed $p$ (or the other way around).

Our first result is very simple yet powerful. Let $\mu$ be a function
satisfying the mild technical conditions that it is polynomial-time
computable and monotone (i.e., if $X\subseteq Y$, then $\mu(X)\le
\mu(Y)$). Observe that for example all three functions defined above
satisfy these conditions. Our first result shows that for {\em every
  function $\mu$} satisfying these conditions and {\em every fixed
  integer} $q$, the problem \partmupqsc\ can be solved in polynomial
time (the value $p$ is considered to be part of the input). For
example, it can be decided in polynomial time if there is a clustering
where at most 13 edges leave each cluster and each cluster induces at
most 27 nonedges (or even the more general question, where 
 the maximum number $p$ of nonedges is given in the input). This
might be  surprising: we believe that most people would guess
that this problem is NP-hard. The algorithm is based on a simple
application of uncrossing of posimodular functions and on the
fact that for fixed $q$ we can enumerate every (connected) cluster with at
most $q$ outgoing edges. The crucial observation is that if every
vertex can be {\em covered} by a good cluster, then the vertices can
be {\em partitioned} into good clusters. Thus the problem boils down
to checking for each vertex $v$ if it is contained in a suitable
cluster.

While the algorithm is simple in hindsight, considerable efforts have
been spent on solving some very particular special cases. For example,
Heggernes et al.~\cite{pq-clustering} gave a polynomial-time algorithm for \partitionsc{\nonedge}{1}{3} and
Langston and Plaut~\cite{DBLP:journals/dm/LangstonP98} argued that the very deep results of Robertson and
Seymour on graph minors and immersions imply that
\partitionsc{\size}{p}{q} is polynomial-time solvable for every fixed
$p$ and $q$. These results follow as straightforward corollaries from our
first result.

Although this simple algorithm is polynomial for every fixed $q$, the
running time on is about $n^{\Theta(q)}$, thus it is not efficient even for small
values of $q$. We do not hope for polynomial time algorithms for the general case, since both the $(\nonedge,p,q)$-\textsc{Partition} and $(\size,p,q)$-\textsc{Partition} problems are known to be NP-complete when both $p$ and $q$ are part of the input~\cite{pq-clustering,Go95}. Hence, to improve the running time, we consider the problem from the viewpoint of parameterized complexity. We show that for
several natural measures $\mu$, including the three defined above, the
clustering problem can be solved in time $2^{O(q)}\cdot
n^{O(1)}$, that is, the problem is fixed-parameter tractable (FPT)
parameterized by the bound $q$ on the number of edges leaving a
cluster. Moreover, the bound $p$ can be assumed to be part of the
input. Thus this algorithm can be efficient for small values of $q$
(say, $O(\log n)$) even if $p$ is large.  The problem
\partitionsc{\size}{p}{q}\ appears in the open problem list of
the 1999 monograph of Downey and Fellows \cite{MR2001b:68042} under the name ``Minimum
Degree Partition,'' where it is suggested that the problem is probably
W[1]-hard parameterized by $q$. Our result answers this question by
showing that the problem is FPT, contrary to the expectation of Downey
and Fellows.

A crucial ingredient of our parameterized algorithm is the notion of
{\em important separators,} which has been used (implicitly or
explicitly) to obtain fixed-parameter tractability results for various
cut- or separator-related problems. In particular, we use the
``randomized selection of important sets'' argument that was
introduced very recently in \cite{marx-razgon-multicut} to prove the
fixed-parameter tractability of (edge and vertex) multicut. With these
tools at hand, we can reduce \partmupqsc\ to a special case that we
call the ``Satellite Problem.'' We show that if the Satellite Problem is
fixed-parameter tractable parameterized by $q$ for a particular
function $\mu$, then \partmupqsc\ is also fixed-parameter tractable
parameterized by $q$. It seems that for many reasonable functions
$\mu$, the Satellite Problem can be solved by dynamic programming
techniques. In particular, this is true for the three functions defined
above, and this results in algorithms with running time
$2^{O(q)}\cdot n^{O(1)}$. Note that the reduction to the \textsc{Satellite
  Problem} works for every monotone $\mu$, and we need arguments
specific to a particular $\mu$ only in the algorithms for \textsc{Satellite
  Problem}.

We also investigate \partmupqsc\ parameterized by $p$ and show that for
$\mu=\size$, $\nonedge$, and $\nondeg$, the problem is FPT
parameterized by $p$: it can be solved in time $2^{O(p)}\cdot
n^{O(1)}$ (this time the value $q$ is part of the input). For these
results, we use a combination of color coding and dynamic programming.
Interestingly, these algorithms rely on the assumption that there are no parallel edges in the graph (in contrast to parameterization by $q$, where our algorithms work the same even if parallel edges are  allowed). In fact, if parallel edges are allowed, then in the case $\mu=\nonedge$ or $\nondeg$, the problem is NP-hard even for $p=0$, while for $\mu=\size$, it is W[1]-hard parameterized by $p$ (i.e., unlikely to be FPT).

Previous work on fixed-parameter tractability of clustering problems
focused mostly on parameterization by the total error. In problems
such as \textsc{Cluster Editing} and \textsc{Cluster Vertex Deletion},
the task is to modify the graph into a disjoint union of cliques by at
most $k$ edge deletions or
additions~\cite{DBLP:conf/fct/FellowsLRS07,DBLP:journals/tcs/Guo09,DBLP:journals/mst/HuffnerKMN10}.
Generalizations of the problem have been considered in
\cite{GuoKKU11,GuoKNU10,DBLP:journals/algorithmica/BevernMN12},
where the graph has to modified in such a way that every component is ``clique-like'' defined by measures similar to the ones in the current paper. It is not possible to directly compare these results
with our results as we explore a different objective: instead of
bounding the total number of operations required to turn the graph
into clusters, we have a bound on the number of operations that can
affect a each cluster. However, in general, we can say that FPT
results are more interesting for parameters that are typically
smaller. Intuitively, the number of editing operations affecting a single
cluster is much smaller than the total number of operations, thus FPT
algorithms for problems parameterized by local bounds on the clusters can be
considered more interesting than FPT algorithms for problems parameterized by the
total number of operations.

\section{Clustering and uncrossing}
\label{sec:uncross}
Given an undirected graph $G$, we denote by $\Delta(X)$ the set of edges between $X$ and $V(G)\setminus
X$, and define $d(X)=|\Delta(X)|$.  We will use two well-known and 
easily checkable properties of the function $d$: for  $X,Y\subseteq
V(G)$, $d$ satisfies the {\em submodular} inequality
\[
 d(X)+d(Y)\ge d(X\cap Y)+d(Y\cup X)
 \]
 and the {\em posimodular} inequality
 \[
 d(X)+d(Y)\ge d(X\setminus Y)+d(Y\setminus X).
 \]

 Let $\mu: 2^{V(G)}\to\mathbb{Z}^+$ be a function
assigning nonnegative integers to sets of vertices of $G$. Let $p$ and
$q$ be two integers. We say that a set $C\subseteq V(G)$ is a {\em
  $(\mu,p,q)$-cluster} if $\mu(C)\le p$ and $d(C)\le q$. A {\em
\partmupq} of $G$ is a partition of $V(G)$ into
$(\mu,p,q)$-clusters. The main problem considered in this paper is
finding such a partition.
A necessary condition for the existence of \partmupq\ is that for
every vertex $v\in V(G)$ there exists a $(\mu,p,q)$-cluster that contains
$v$. Therefore, we are also interested in the problem of finding a
cluster that contains a particular vertex.

\noindent
\parbox{0.45\linewidth}{
\problemqz{\partmupqsc}{A graph $G$, integers $p$, $q$.}
{A \partmupq\ of $G$.}}
\parbox{0.55\linewidth}{
\problemqz{$(\mu,p,q)$-\textsc{cluster}}
{Graph $G$, integers $p$, $q$, vertex $v$.}
{A $(\mu,p,q)$-cluster $C$ containing $v$.}}

The main observation of this section is that if $\mu$ is {\em
  monotone} (i.e., $\mu(X)\le \mu(Y)$ for every $X\subseteq Y$), then
every vertex $v$ being in some cluster is actually a sufficient
condition. Therefore, in these cases, it is sufficient to solve
\textsc{$(\mu,p,q)$-cluster}.

\begin{lemma}\label{lem:uncrossing}
  Let $G$ be a graph, let $p,q\ge 0$ be two integers, and let
  $\mu:2^{V(G)}\to \mathbb{Z}^+$ be a 
  monotone function. If every $v\in V(G)$ is contained in some
  $(\mu,p,q)$-cluster, then $G$ has a \partmupq.
  Furthermore, given a set of $(\mu,p,q)$-clusters $C_1$, $\dots$,
  $C_n$ whose union is $V(G)$, a \partmupq\ can be found
  in polynomial time.
\end{lemma}

\begin{proof}
Let us consider a collection $C_1$, $\dots$, $C_n$ of
$(\mu,p,q)$-clusters whose union is $V(G)$. If the sets are pairwise
disjoint, then they form a partition of $V(G)$ and we are done. If
$C_i\subseteq C_j$, then the union remains $V(G)$ even after throwing
away $C_i$. Thus we can assume that no set is contained in
another. Suppose that $C_i$ and $C_j$ intersect. Now either
$d(C_i)\ge d(C_i\setminus C_j)$ or $d(C_j)\ge
d(C_j\setminus C_i)$ must be true: it is not possible that both
$d(C_i)<d(C_i\setminus C_j)$ and  $d(C_j)<
d(C_j\setminus C_i)$ hold, as this would violate the
posimodularity of $d$. Suppose that $d(C_j)\ge
d(C_j\setminus C_i)$. Now the set $C_j\setminus C_i$ is also a
$(\mu,p,q)$-cluster: we have $d(C_j\setminus C_i)\le
d(C_j)\le q$ by assumption and $\mu(C_j\setminus C_i)\le
\mu(C_j)\le p$ from the monotonicity of $\mu$. Thus we can replace
$C_j$ by $C_j\setminus C_i$ in the collection: it will remain true
that the union of the clusters is $V(G)$. Similarly, if
$d(C_i)\ge d(C_i\setminus C_j)$, then we can replace $C_i$
by $C_i\setminus C_j$.

Repeating these steps (throwing away subsets and resolving
intersections), we eventually arrive at a pairwise disjoint collection
of $(\mu,p,q)$-clusters. Each step decreases the number of cluster pairs $C_i$, $C_j$ that have non-empty intersection.
Therefore, this process terminates after a polynomial number of steps.
\end{proof}

The proof of Lemma~\ref{lem:uncrossing} might suggest that we can
obtain a partition by simply taking, for every vertex $v$, a
$(\mu,p,q)$-cluster $C_v$ that is inclusionwise minimal with respect
to containing $v$. However, such clusters can still cross. For
example, consider a graph on vertices $a$, $b$, $c$, $d$ where every
pair of vertices expect $a$ and $d$ are adjacent. Suppose that
$\mu=\size$, $p=3$, $q=2$. Then $\{a,b,c\}$ is a minimal cluster
containing $b$ (as more than two edges are going out of each of
$\{b\}$, $\{b,c\}$, and $\{a,b\}$) and $\{b,c,d\}$ is a minimal
cluster containing $c$. Thus unless we choose the minimal clusters
more carefully in a coordinated way, they are not guaranteed to form a
partition. In other words, there are two symmetric solutions
$(\{a,b,c\},\{d\})$ and $(\{a\},\{b,c,d\})$ for the problem, and the
clustering algorithm has to break this symmetry somehow.

In light of Lemma~\ref{lem:uncrossing}, it is sufficient to find a
$(\mu,p,q)$-cluster $C_v$ for each vertex $v\in V(G)$. If there is a
vertex $v$ for which there is no such cluster $C_v$, then obviously
there is no \partmupq; if we have such a $C_v$ for every
vertex $v$, then Lemma~\ref{lem:uncrossing} gives us a
\partmupq\ in polynomial time. For fixed $q$,
\textsc{$(\mu,p,q)$-Cluster} can be  solved by brute force if $\mu$ is
polynomial-time computable: enumerate every set $F$ of at most $q$
edges and check if the component of $G\setminus F$ containing $v$ is a
$(\mu,p,q)$-cluster. If $C_v$ is a $(\mu,p,q)$-cluster containing $v$, then we 
find it when $F=\Delta(C_v)$ is considered by the enumeration procedure.
\begin{theorem}\label{th:brute}
  Let $\mu$ be a polynomial-time computable monotone function. Then
  for every fixed $q$, there is an $n^{O(q)}m$ time algorithm for
  \partmupqsc.
\end{theorem}
As we have seen, an algorithm for \textsc{$(\mu,p,q)$-Cluster} gives
us an algorithm for \partmupqsc. In the rest of
the paper, we devise more efficient algorithms for \textsc{$(\mu,p,q)$-Cluster} than the
$n^{O(q)}$ time brute force method described above.

\section{Parameterization by $q$}
\label{sec:parq}
The main result of this section is that
\partmupqsc\ is  fixed-parameter tractable
parameterized by $q$ for the
three functions \nonedge, \nondeg, and \size.
\begin{theorem}\label{lem:thmainq} 
  \partitionsc{\size}{p}{q}, \partitionsc{\nonedge}{p}{q}, and
  \partitionsc{\nondeg}{p}{q} can be solved in time  $2^{O(q)}n^{O(1)}$.
\end{theorem}
 By
Lemma~\ref{lem:uncrossing}, all we need to show is that
\textsc{$(\mu,p,q)$-cluster} is fixed-parameter
tractable parameterized by $q$. We introduce a somewhat technical
variant of this question, the \textsc{Satellite Problem}, and show that for {\em
  every monotone function $\mu$},
 if \textsc{Satellite Problem} is FPT, then
 \textsc{$(\mu,p,q)$-cluster} is FPT as well.
Thus we need arguments specific to a particular $\mu$ only in solving
the \textsc{Satellite Problem}.

\problemq{\textsc{Satellite Problem}}
{A graph $G$, integers $p$, $q$, a vertex $v\in V(G)$, a partition $V_0$, $V_1$, $\ldots$,
  $V_r$ of $V(G)$ such that $v\in V_0$ and there is no edge between
  $V_i$ and $V_j$ for any $1\le i < j \le r$. }
{A $(\mu,p,q)$-cluster $C$ with $V_0\subseteq C$ such that for every $1\le i \le r$, either $C\cap V_i=\emptyset$ or $V_i\subseteq C$.}

\begin{figure}
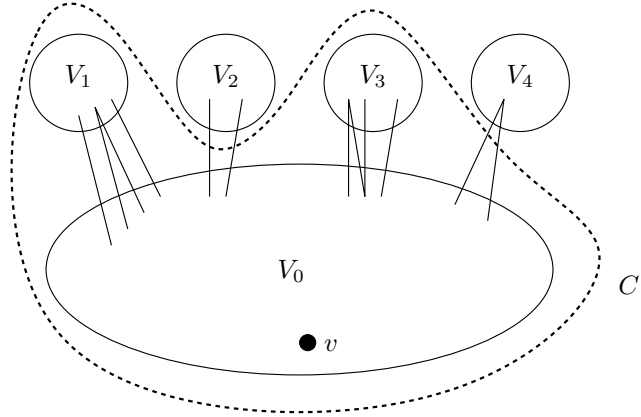

\begin{center}
\svg{0.5\linewidth}{star}
\caption{Instance of \textsc{Satellite Problem} with a solution $C$.}\label{fig:star}
\end{center}
\end{figure}
Since the sets $\{V_i\}$ form a partition of $V(G)$, we have $r \leq n$. For every $V_i$ ($1\leq i \leq r$), we have to decide whether to include or exclude it from the solution $C$ (see Fig.~\ref{fig:star}). If we exclude $V_i$ from $C$, then $d(C)$ increases by the number of edges between $V_0$ and $V_i$. If we include $V_i$ into $C$, then $\mu(C)$ increases accordingly. Thus we
need to solve the knapsack-like problem of including sufficiently many
$V_i$ such that $d(C)\le q$, but not including too many to ensure
$\mu(C)\le p$. As we shall see in
Section~\ref{sec:solving-star-problem}, in many cases this problem can
be solved by dynamic programming (and some additional arguments). The
important fact that we use is that there are no edges between $V_i$
and $V_j$, thus for many reasonable functions $\mu$, the way $\mu(C)$ increases
by including $V_i$ is fairly independent from whether $V_j$ is
included in $C$ or not.

The reduction to the \textsc{Satellite Problem} uses the concept of
important separators (Section~\ref{sec:important-separators}) and in
particular the technique of ``randomly selecting important
separators'' introduced in \cite{marx-razgon-multicut}. As the
reduction can be most conveniently described as a randomized
algorithm, we present it this way in
Section~\ref{sec:reduct-star-probl}, and then show how it can be
derandomized in Section~\ref{sec:derand}. In
Section~\ref{sec:solving-star-problem}, we show how the
\textsc{Satellite Problem} can be solved for the three functions
\nonedge, \nondeg{} and \size.

\subsection{Important separators}
\label{sec:important-separators}
The notion of {\em important separators} was introduced in
\cite{marx-separation-full} to prove the fixed-parameter tractability
of multiway cut problems. This notion turned out to be useful in other applications as well \cite{ChenLL09,dblp:journals/jacm/chenllor08,DBLP:journals/jcss/RazgonO09}. The basic idea is that in many problems where terminals need to be separated in some way, it is sufficient to consider separators that are ``as far as possible'' from one of the terminals.

Since there are some small differences between edge and vertex separators, and some of the results appear only implicitly in previous papers, we make the paper self-contained by restating all the definitions and by reproving all the required results in this section. Let $s,t$ be two vertices of a graph $G$. An {\em $s-t$ separator} is a set $S\subseteq E(G)$ of edges separating $s$ and $t$, i.e., there is no $s-t$ path in $G\setminus S$. An $s-t$ separator is {\em inclusionwise minimal} if there is an $s-t$ path in $G\setminus S'$ for every $S'\subset S$.
\begin{definition}
Let $s,t\in V(G)$ be vertices, $S\subseteq E(G)$ be an $s-t$ separator, and let $K$ be the component of $G\setminus S$ containing $s$. We say that $S$ is an {\em important $s-t$ separator} if it is inclusionwise minimal and there is no $s-t$ separator $S'$ with $|S'|\le |S|$ such that $K\subset K'$ for the component $K'$ of $G\setminus S'$ containing $s$.
\end{definition}

Note that an important $s-t$ separator is not necessarily an important
$t-s$ separator. Intuitively, we want to minimize the size of the
$s-t$ separator and at the same time we want to maximize the set of
vertices that remain reachable from $s$ after removing the
separator. The important separators are the solutions that are
Pareto-optimal with respect to these two objectives. Note that we do
not want that the number of vertices reachable from $s$ to be maximal,
we just want that this set of vertices is {\em inclusionwise maximal}
(i.e., we have $K\subset K'$ and {\em not} $|K|<|K'|$ in the
definition). The main observation of \cite{marx-separation-full} is
that the number of important $s-t$ separators of size at most $k$ can
be bounded by a function of $k$.

\begin{theorem}\label{th:impsep}
Let $s,t\in V(G)$ be two vertices in graph $G$. For every $k\ge 0$, there are at most $4^k$ important $s-t$ separators of size at most $k$. Furthermore, these important $s-t$ separators can be enumerated in time $4^k\cdot n^{O(1)}$.
\end{theorem}
The following lemma clearly proves the bound in Theorem~\ref{th:impsep}: if the sum is at most 1, then there cannot be more than $4^k$ important $s-t$ separators of size at most $k$. Although this form of the statement is new, the proof follows the same ideas that appear implicitly in~\cite{dblp:journals/jacm/chenllor08,ChenLL09}.

\begin{lemma}\label{lem:impsepsum}
Let $s,t\in V(G)$. If $\S$ is the set of
all important $s-t$ separators, then $\sum_{S\in \S}4^{-|S|}\le
1$. Thus $\S$ contains at most $4^k$ $s-t$ separators of size at most $k$.
\end{lemma}
\begin{proof}
  Let $\lambda$ be the size of the smallest $s-t$ separator. We prove
  by induction on the number of edges of $G$ that $\sum_{S\in
    \S}4^{-|S|}\le 2^{-\lambda}$. If $\lambda=0$, then there is a
  unique important $s-t$ separator of size at most $k$: the empty set.
  Thus we can assume that $\lambda>0$.

  For an $s-t$ separator $S$, let $K_S$ be the component of
  $G\setminus S$ containing $s$ (e.g., if $S$ is an inclusionwise
  minimal $s-t$ separator, then $S=\Delta(K)$). First we show the well-known fact
  that there is a unique $s-t$ separator $S^*$ of size $\lambda$ such
  that $K_{S^*}$ is inclusionwise maximal, i.e., there is no other
  $s-t$ separator $S$ of size $\lambda$ with $K_{S^*}\subset K_S$.
  Suppose that there are two separators $S'$ and $S''$ with
  $|S'|=|S''|=\lambda$ that are
  inclusionwise maximal in this sense. By the submodularity of
  $d$, we have
\begin{align*}
  \underbrace{d(K_{S'})}_{=\lambda}+\underbrace{d(K_{S''})}_{=\lambda}\ge
d(K_{S'}\cup K_{S''})+
\underbrace{d(K_{S'}\cap K_{S''})}_{\ge \lambda}.
\end{align*}
The left hand side is exactly $2\lambda$, while the second term of the
right hand side is at least $\lambda$ (as $\Delta(K_{S'}\cap
K_{S''})$ is an $s-t$-separator). Therefore,
$d(K_{S'} \cup K_{S''})\le \lambda$. This means that
$\Delta(K_{S'}\cup K_{S''})$ is also a minimum $s-t$ cut, contradicting the
maximality of $S'$ and $S''$.

Next we show that for every important $s-t$ separator $S$, we have
$K_{S^*}\subseteq K_S$. Suppose this is not true for some $S$. We use
submodularity again:
\[
\underbrace{d(K_{S^*})}_{=\lambda}+d(K_{S})\ge
d(K_{S^*}\cup K_{S})+
\underbrace{d(K_{S^*}\cap K_{S})}_{\ge \lambda}.
\]
By definition, $d(K_{S^*})=\lambda$, and $\Delta(K_{S^*}\cap K_{S})$
is an $s-t$ separator, hence $d(K_{S^*}\cap K_{S})\ge
\lambda$. This means that $d(K_{S^*}\cup K_{S})\le
d(K_{S})$. However this contradicts the assumption that $S$ is an
important $s-t$ separator: $\Delta(K_{S^*}\cup K_{S})$ is an $s-t$
separator not larger than $S$, but $K_{S^*}\cup K_{S}$ is a proper
superset of $K_S$ (as $K_{S^*}$ is not a subset of $K_S$ by
assumption).

We have shown that for every important separator $S$, the set $K_S$
contains $K_{S^*}$. Let $e\in S^*$ be an arbitrary edge of $S^*$ (note
that $\lambda>0$, hence $S^*$ is not empty) and let $v$ be the
endpoint of $e$ not in $K_{S^*}$. An important $s-t$ separator $S$
either contains $e$ or not. We will bound the contributions of these
two types of separators to the sum.

Let $S$ be an important $s-t$ separator containing $e$. Then $S\setminus e$ is an $s-t$ separator in $G\setminus
e$; in fact, it is an important $s-t$
separator of $G\setminus e$. 
Therefore, if $\S'$ is the set of all important $s-t$ separators in
$G\setminus e$, then the set $\S_1=\{S'\cup e\mid S'\in \S'\}$ contains every
important $s-t$ separator of $G$ containing $e$. 
Obviously, the size $\lambda'$ of the
minimum $s-t$ separator in $G\setminus e$ is at least $\lambda-1$.
As $G\setminus e$ has fewer edges than $G$, the induction statement shows that $\sum_{S'\in
  \S'}4^{-|S'|}\le  2^{-\lambda'}\le 2^{-(\lambda-1)}$ and therefore $\sum_{S\in
  \S_1}4^{-|S|}=\sum_{S'\in \S'}4^{-|S'|-1}\le 2^{-(\lambda-1)}/4=2^{-\lambda}/2$.

Let us consider now the important $s-t$ separators not containing $e$.
We have seen that $K_{S^*}\subseteq K_S$ for every such $s-t$
separator $S$. As $e\not \in S$, even $K_{S^*}\cup\{v\}\subseteq K_S$
is true. Let us obtain the graph $G'$ from $G$ by removing
$(K_{S^*}\cup\{v\})\setminus \{s\}$ and making $s$ adjacent to the
neighborhood of $K_{S^*}\cup \{v\}$ in $G$ (or equivalently, by
contracting $K_{S^*}\cup \{v\}$ into $s$).  Note that $G'$ has
strictly fewer edges than $G$.  There is no $s-t$ separator $S$ of
size $\lambda$ in $G'$: such a set $S$ would be an $s-t$ separator of
size $\lambda$ in $G$ as well, with $K_{S^*}\cup \{v\}\subseteq K_S$,
contradicting the maximality of $S^*$. Thus the minimum size
$\lambda'$ of an $s-t$ separator in $G'$ is strictly greater than
$\lambda$.  Let $\S_2$ contain all the important $s-t$ separators of
$G$ not containing $e$. We have seen that $K_{S^*}\cup \{v\}\subseteq
K_S$ for every separator $S\in \S_2$, thus such an $S$ is an $s-t$
separator of $G'$ and in fact every such $S$ is an important $s-t$
separator in $G'$ as well. Therefore, by the induction hypothesis,
$\sum_{S\in \S_2}4^{-|S|}\le 2^{-\lambda'} \le 2^{-\lambda}/2$.
Adding the bounds in the two cases, we get the required bound $2^{-\lambda}$.
\end{proof}
Note that the proof of Lemma~\ref{lem:impsepsum} gives a branching
procedure for enumerating all the important separators of a certain
size. This proves the algorithmic claim in Theorem~\ref{th:impsep}: as
each branching step can be performed in polynomial time, the bound on
the running time follows from the bound on the number of important
separators.

\subsection{Reduction to the Satellite Problem}
\label{sec:reduct-star-probl}

In this section we show how to reduce $(\mu,p,q)$-\textsc{Cluster} to
the \textsc{Satellite Problem} by a randomized reduction
(Lemma~\ref{lem:redstar}). Section~\ref{sec:derand} shows that the
same result can be obtained by a deterministic algorithm as well
(Lemmas \ref{lem:redstarderand} and
\ref{lem:redstarderand2}). However, the randomized version is
conceptually simpler, thus we present it first and then discuss the
derandomization in the next section.
\begin{lemma}\label{lem:redstar}
If \textsc{Satellite Problem} can be solved in time $f(q)\cdot n^{O(1)}$
for some monotone $\mu$, then there is a randomized $2^{O(q)}\cdot f(q)\cdot
n^{O(1)}$ algorithm with constant error probability that finds a
$(\mu,p,q)$-cluster containing $v$ (if one exists).
\end{lemma}
The crucial definition of this section is the following:
\begin{definition}\label{def:impset}
We say that a set $X\subseteq V(G)$, $v\not\in X$ is {\em important} if
\begin{enumerate}
\item $d(X)\le q$,
\item $G[X]$ is connected,
\item there is no $Y\supset X$, $v\not\in Y$ such that $d(Y)\le d(X)$ and
  $G[Y]$ is connected.
\end{enumerate}
\end{definition}
It is easy to see that $X$ is an important set if and only if
$\Delta(X)$ is an important $u-v$ separator of size at most $q$ for
every $u\in X$. Thus we can use Theorem~\ref{th:impsep} to enumerate
every important set, and Lemma~\ref{lem:impsepsum} to give an upper
bound the number of important sets. Lemma~\ref{lem:impcomponent}
establishes the connection between important sets and finding
$(\mu,p,q)$-clusters: we can assume that the components of $G\setminus
C$ for the solution $C$ are important sets. In
Lemma~\ref{lem:redstarrand}, we show that by randomly choosing
important sets, with some probability we can obtain an instance of the
\textsc{Satellite Problem} where $V_1$, $\dots$, $V_r$ includes all the
components of $G\setminus C$. This gives us the reduction 
 stated in Lemma~\ref{lem:redstar} above.

\begin{lemma}\label{lem:impcomponent}
Let $C$ be an inclusionwise minimal $(\mu,p,q)$-cluster containing
$v$. Then every component of $G\setminus C$ is an important set.
\end{lemma}
\begin{proof}
  Let $X$ be a component of $G\setminus C$. It is clear that $X$
  satisfies the first two properties of Definition~\ref{def:impset}
  (note that $\Delta(X)\subseteq \Delta(C))$. Thus let us suppose that
  there is a $Y\supset X$, $v\not\in Y$ such that $d(Y)\le
  d(X)$ and $G[Y]$ is connected. Let $C':=C\setminus Y$. Note
  that $C'$ is a proper subset of $C$: every neighbor of $X$ is in
  $C$, thus a connected superset of $X$ has to contain at least one vertex of
  $C$. It is easy to see that $C'$ is a $(\mu,p,q)$-cluster: we have
  $\Delta(C')\subseteq(\Delta(C)\setminus \Delta(X))\cup \Delta(Y)$ and
  therefore $d(C')\le d(C)-d(X)+d(Y)\le
  d(C)\le q$ and $\mu(C')\le \mu(C)\le p$ (by the monotonicity of
  $\mu$). This contradicts the minimality of $C$.
\end{proof}

\begin{lemma}\label{lem:redstarrand}
Given a graph $G$, vertex $v\in V(G)$, integers $p$, $q$, and a monotone function
$\mu:2^{V(G)}\to \mathbb{Z}^+$, we can construct in time
$2^{O(q)}\cdot n^{O(1)}$ an instance $I$ of the
\textsc{Satellite Problem} such that
\begin{itemize}
\item If some $(\mu,p,q)$-cluster contains $v$, then $I$ is a
  yes-instance with probability $2^{-O(q)}$,
\item If there is no $(\mu,p,q)$-cluster containing $v$, then $I$ is a no-instance.
\end{itemize}
\end{lemma}
\begin{proof}
For every $u\in V(G)$, $u\neq v$, let us use the algorithm of 
Theorem~\ref{th:impsep} to enumerate every important
$u-v$ separator of size at most $q$. For every such separator $S$, let
us put the component $K$ of $G\setminus S$ containing $u$ into the
collection $\X$. Note that the same component $K$ can be obtained for more
than one vertex $u$, but we put only one copy into $\X$.

Let $\X'$ be a subset of $\X$, where each member $K$ of $\X$ is chosen
with probability $4^{-d(K)}$ independently at random. Let $Z$ be
the union of the sets in $\X'$, let $V_1$, $\dots$, $V_r$ be the
connected components of $G[Z]$, and let $V_0=V(G)\setminus Z$. It is
clear that $V_0$, $V_1$, $\dots$, $V_r$ give an instance $I$ of the
\textsc{Satellite Problem}, and a solution for $I$ gives a
$(\mu,p,q)$-cluster containing $v$. Thus we only need to show that if
there is a $(\mu,p,q)$-cluster $C$ containing $v$, then $I$ is a
yes-instance with probability $2^{-O(q)}$.

Let $C$ be an inclusionwise minimal $(\mu,p,q)$-cluster containing $v$.
Let $S$ be the set of vertices on the boundary of $C$, i.e., the vertices of
$C$ incident to $\Delta(C)$. Let $K_1$, $\dots$, $K_t$ be the
components of $G\setminus C$. Note that every edge of $\Delta(C)$
enters some $K_i$, thus $\sum_{i=1}^t d(K_i)=d(C)\le q$. By
Lemma~\ref{lem:impcomponent}, every $K_i$ is an important set, and
hence it is in $\X$. Consider the following two events:
\begin{enumerate}
\item[(E1)] Every component $K_i$ of $G\setminus C$ is in $\X'$ (and
  hence $K_i\subseteq Z$).
\item[(E2)] $Z\cap S=\emptyset$.
\end{enumerate}
The probability that (E1) holds is $\prod_{i=1}^t
4^{-d(K_i)}=4^{-\sum_{i=1}^td(K_i)}\ge 4^{-q}$. Event (E2)
holds if for every $w\in S$, no set $K\in \X$ with $w\in K$ is selected into
$\X'$. It follows directly from the definition of important separators that for every $K\in \X$ with $w\in K$, $\Delta(K)$ is
an important $w-v$ separator. Thus by Lemma~\ref{lem:impsepsum},
$\sum_{K\in \X, w\in K}4^{-|d(K)|}\le 1$. The probability that
$Z\cap S=\emptyset$ can be bounded by 
{\small \begin{multline*}
\prod_{K\in \X,K\cap S\neq\emptyset}(1-4^{-d(K)})\ge
\prod_{w\in S}\prod_{K\in \X, w\in K} (1-4^{-d(K)})\ge 
\prod_{w\in S}\prod_{K\in \X, w\in K} \exp\left(\frac{-
4^{-d(K)}}{(1-4^{-d(K)})}\right)\\\ge
\prod_{w\in S}\prod_{K\in \X, w\in K} \exp\left(-\frac{4}{3}\cdot 4^{-d(K)}\right)
=\prod_{w\in S}\exp\left(-\frac{4}{3}\cdot \sum_{K\in \X, w\in K}4^{-d(K)}\right)\ge 
(e^{-\frac{4}{3}})^{|S|}\ge e^{-4q/3}.
\end{multline*}} In the first inequality, we use that every term is
less than 1 and every term on the right hand side appears at least
once on the left hand side; in the second inequality, we use that
$1+x\ge \exp(x/(1+x))$ for every $x>-1$. Events (E1) and (E2) are
independent: (E1) is a statement about the selection of a subcollection
$A\subseteq \X$ of at most $q$ sets that are disjoint from $S$, while
(E2) is a statement about not selecting any member of a subcollection
$B\subseteq \X$ of at most $|S|\cdot 4^q$ sets intersecting $S$. Thus by
probability $2^{-O(q)}$, both (E1) and (E2) hold.

Suppose that both (E1) and (E2) hold, we show that instance $I$ of the
\textsc{Satellite Problem} is a yes-instance. In this case, every component
$K_i$ of $G\setminus C$ is a component $V_j$ of $G[Z]$:
$K_i\subseteq Z$ by (E1) and every neighbor of $K_i$ is outside $Z$. Thus $C$ is
a solution of $I$, as it can be obtained as the union of $V_0$ and
some components of $G[Z]$.
\end{proof}


\subsection{Derandomization of the Reduction to the Satellite Problem}
\label{sec:derand}

To derandomize the proof of Lemma~\ref{lem:redstarrand} and obtain a
deterministic version of Lemma~\ref{lem:redstar}, we use the standard
technique of splitters. A {\em $(n,k,k^2)$-splitter} is a
family of functions from $[n]$ to $[k^2]$ such that for any subset
$X\subseteq [n]$ with $|X|=k$, one of the functions in the family is
injective on $X$. Naor, Schulman, and Srinivasan \cite{NSS95} gave an
explicit construction of an $(n,k,k^2)$-splitter of size $O(k^6\log k
\log n)$.

First we present a simpler version of the derandomization
(Lemma~\ref{lem:redstarderand}), where the dependence on $q$ is
$2^{O(q^2)}$ (instead of the $2^{O(q)}$ dependence of the randomized
algorithm). The derandomization is along the same lines as the
analogous proof in \cite{marx-razgon-multicut}. Then we improve the
dependence to $2^{O(q)}$ by a somewhat more complicated scheme and
analysis (Lemma~\ref{lem:redstarderand2}). 

\begin{lemma}\label{lem:redstarderand}
If \textsc{Satellite Problem} can be solved in time $f(q)\cdot n^{O(1)}$
for some monotone $\mu$, then there is a $2^{O(q^2)}\cdot f(q)\cdot
n^{O(1)}$ algorithm for \textsc{$(\mu,p,q)$-Cluster}.
\end{lemma}

\begin{proof}
  In the algorithm of Lemma~\ref{lem:redstarrand}, a random subset of
  a universe $\X$ of size $s=|\X|\le 4^q\cdot n$ is selected. If
  the \textsc{$(\mu,p,q)$-Cluster} problem has a solution $C$, then
  there is a collection $A\subseteq \X$ of at most $a=q$ sets and a
  collection $B\subseteq \X$ of at most $b=q\cdot 4^q$ sets such that
  if every set in $A$ is selected and no set in $B$ is selected, then
  (E1) and (E2) hold.  Instead of selecting a random subset, we
  try every function $f$ in an $(s,a+b,(a+b)^2)$-splitter family and
  every subset $F\subseteq [(a+b)^2]$ of size $a$ (there are
  $\binom{(a+b)^2}{a}=2^{O(q^2)})$ such sets $F$). For a particular
  choice of $f$ and $F$, we select those sets $S\in \X$ into $\X'$ for
  which $f(S)\in F$.  The size of the splitter family is $2^{O(q)}\log
  n$ and the number of possibilities for $F$ is $2^{O(q^2)}$.
  Therefore, we construct $2^{O(q^2)}\cdot\log n$ instances of
  the \textsc{Satellite Problem}.

  By the definition of the splitter, there will be a function $f$ that
  is injective on $A\cup B$, and there is a subset $F$ such that
  $f(S)\in F$ for every set $S$ in $A$ and $f(S)\not\in F$ for every
  set $S$ in $B$. For such an $f$ and $F$, the selection will ensure
  that (E1) and (E2) hold.  This means that the constructed instance of
  the \textsc{Satellite Problem} corresponding to $f$ and $F$ has a
  solution as well. Thus solving every constructed instance of the
  \textsc{Satellite Problem} with the assumed $f(q)\cdot n^{O(1)}$ time
  algorithm gives a $2^{O(q^2)}\cdot f(q)\cdot n^{O(1)}$ algorithm for
  \textsc{$(\mu,p,q)$-Cluster}.
\end{proof}

The key modification that we need in order to improve the dependence on $q$ is
to do the selection of sets with different boundary sizes separately,
and use a separate splitter for each boundary size. This modification
makes the analysis of the running time much more complicated.
\begin{lemma}\label{lem:redstarderand2}
If \textsc{Satellite Problem} can be solved in time $f(q)\cdot n^{O(1)}$
for some monotone $\mu$, then there is a $2^{O(q)}\cdot f(q)\cdot
n^{O(1)}$ algorithm for \textsc{$(\mu,p,q)$-Cluster}.
\end{lemma}
\begin{proof}
Let the universe $\X$, the fixed solution $C$, and the collections $A$
and $B$ be as in the proof of Lemma~\ref{lem:redstarderand}. Let $\X_i=\{K\in \X\mid d(K)=i\}$
and let $a_i=|\X_i\cap A|$, i.e., the number of sets $K\in A$ that
have $i$ edges on its boundary. Observe that $a_i=0$ for $i>q$ and
$\sum_{i=1}^{q}a_i\cdot i=d(C)\le q$. In the first step of the
algorithm, we guess the values $a_1$, $\dots$, $a_q$ that correspond
to the fixed hypothetical solution $C$. The number of possibilities for
these values can be bounded by $2^{O(q)}$ (this is already true if we
have only the weaker requirement $\sum_{i=1}^{q}a_i\le q$). Therefore,
the algorithm branches into $2^{O(q)}$ directions and this guess
introduces only a factor of $2^{O(q)}$ into the running time. From now
on, we assume that we have the correct values of $a_i=|\X_i\cap A|$
corresponding to $C$. We do not know the size of $\X_i\cap B$, but
$b_i=q\cdot 4^i$ is an upper bound on $|\X_i\cap B|$: the set $C$ has
at most $q$ boundary vertices, and each vertex is contained in at most $4^i$
sets of $\X$ (see the proof of Theorem~\ref{lem:redstarrand}).

We perform the selection separately for each $\X_i$ for which $a_i\neq
0$ (if $a_i=0$, then it is safe not to select any member of
$\X_i$). For a particular $\X_i$, we proceed similarly to the 
simplified proof of Lemma~\ref{lem:redstarderand}. That is, for every $1\le i \le q$, we construct an
$(|\X_i|,a_i+ b_i,(a_i+ b_i)^2)$-splitter family $\F_i$ and try every
choice of a function $f_i\in \F_i$ and a subset $F_i\subseteq
[(a_i+ b_i)^2]$ of size $a_i$. For a given choice of $f_1$, $\dots$,
$f_q$ and $F_1$, $\dots$, $F_q$, we select a set $K\in \X_i$ if and
only if $f_i(K)\in F_i$. As in the previous proof, it is clear that at
least one choice of the $f_i$'s and $F_i$'s leads to the selection of
every member of $A$ without selecting any member of $B$.

To bound the running time of the algorithm, we need to bound the total
number of possibilities for $f_i$'s and $F_i$'s. The family $\F_i$ has
size $(a_i+ b_i)^{O(1)}\log n$ and the number of possibilities for $F_i$ is $\binom{a_i+ b_i}{a_i}$. Therefore, we need to show that
\begin{equation}
\prod_{\substack{1\le i \le q\\a_i\neq 0}}(a_i+ b_i)^{O(1)}\cdot \log n\cdot \binom{a_i+ b_i}{a_i}
\label{eq:derand}
\end{equation}
can be bounded by $2^{O(q)}\cdot n^{O(1)}$.

We bound the product of the three factors in \eqref{eq:derand}
separately. Note that it follows from $\sum_{i=1}^qa_i\cdot i\le q$
that $a_i$ can be nonzero for at most $O(\sqrt{q})$ values of $i$.
Therefore, the product of the first factor in \eqref{eq:derand} can
be bounded by
\begin{multline*}
\prod_{\substack{1\le i \le q\\a_i\neq 0}}(a_i+ b_i)^{O(1)}\le 
\prod_{\substack{1\le i \le q\\a_i\neq 0}}(2a_i b_i)^{O(1)}\le 
\prod_{\substack{1\le i \le q\\a_i\neq 0}}(2\cdot 2^{a_i}\cdot q\cdot 4^i)^{O(1)}\\\le 
2^{O(\sqrt{q})}\cdot q^{O(\sqrt{q})}\cdot \prod_{1\le i \le q} (2^{a_i\cdot i}\cdot 4^{a_i\cdot i})\le 2^{O(q)}
\end{multline*}
(in the last inequality, we used $\sum_{1}^qa_i\cdot i\le q$).
To bound the product of the second factor in \eqref{eq:derand}, we consider two cases. If $\log n\le 2^{\sqrt{q}}$, then
\[
\prod_{\substack{1\le i \le q\\a_i\neq 0}}\log n \le \log^{O(\sqrt{q})} n \le 2^{O(q)}.
\]
Otherwise, if $\log n>2^{\sqrt{q}}$, then $\sqrt{q}<\log\log n$, and hence $\log^{O(\sqrt{q})}n=2^{O(\sqrt{q}\cdot \log \log n)}<2^{O( (\log \log n)^2)}=O(n)$.

Finally, let us bound the products of the last factor in \eqref{eq:derand}. Note that $a_i \le q \le b_i$. Therefore, we have 
\begin{multline*}
\prod_{\substack{1\le i \le q\\a_i\neq 0}}\binom{a_i+b_i}{a_i}\le 
\prod_{\substack{1\le i \le q\\a_i\neq 0}}\binom{2b_i}{a_i}\le 
\prod_{\substack{1\le i \le q\\a_i\neq 0}} \left(\frac{2eb_i}{a_i}\right)^{a_i}=
\\
\prod_{\substack{1\le i \le q\\a_i\neq 0}} \left(2e\cdot 4^i\cdot \frac{q}{a_i}\right)^{a_i}
=2^{O(q)}\cdot
4^{\sum_{i=1}^q a_i\cdot i}\cdot
\prod_{\substack{1\le i \le q\\a_i\neq 0}}(q/a_i)^{a_i}
=2^{O(q)}\cdot
\prod_{\substack{1\le i \le q\\a_i\neq 0}}(q/a_i)^{a_i}.
\end{multline*}
Thus we need to bound only $\prod_{i=1}^{q}(q/a_i)^{a_i}$.  For
notational convenience, let $x_i=q/a_i$ whenever $a_i\neq 0$. We bound
separately the product of terms with $x_i\le e^{i}$ and $x_i>e^i$. In
the first case,
\[
\prod_{\substack{1\le i \le q\\ x_i\le e^i}} (q/a_i)^{a_i}=\prod_{\substack{1\le i \le q\\ x_i\le e^i}}x_i^{a_i}
\le \prod_{\substack{1\le i \le q\\ x_i\le e^i}} e^{a_i\cdot i}
\le \exp\left(\sum_{i=1}^q {a_i\cdot i}\right)
=2^{O(q)}.
\]
We use the fact that the function $x^{1/x}$ is monotonically decreasing for $x\ge e$. Therefore, if $x_i>e^i$, then
\[
\prod_{\substack{1\le i \le q\\ x_i> e^i}} (q/a_i)^{a_i}=\prod_{\substack{1\le i \le q\\ x_i> e^i}}\left(x_i^{1/x_i}\right)^{q}
< \prod_{\substack{1\le i \le q\\ x_i> e^i}} \left((e^i)^{1/e^i}\right)^{q}
\le
\exp\left(q\cdot \sum_{i=1}^q i/e^i\right)=2^{O(q)}.
\]
\end{proof}


\subsection{Solving the Satellite Problem}
\label{sec:solving-star-problem}
In this section, we give efficient algorithms for solving the {\sc Satellite Problem} when the function $\mu$ is \size, \nonedge{} and \nondeg. We describe the three algorithms by increasing difficulty. In the case when $\mu$ is \size{}, solving the {\sc Satellite Problem} turns out to be equivalent to the classical {\sc Knapsack} problem with polynomial bounds on the values and weights of the items. 

Recall that the input to the {\sc Satellite Problem} is a graph $G$,
integers $p$, $q$, a vertex $v\in V(G)$, a partition $V_0$, $V_1$,
$\dots$, $V_r$ of $V(G)$ such that $v\in V_0$ and there is no edge
between $V_i$ and $V_j$ for any $1\le i < j \le r$. We denote by $n$
and $m$ the number of vertices and edges of $G$, respectively. The
task is to find a vertex set $C$, such that $C = V_0 \cup \bigcup_{i
  \in S} V_i$ for a subset $S$ of $\{1,\ldots,r\}$ and $C$ satisfies
$d(C) \leq q$ and $\mu(C) \leq p$. For a subset $S$ of
$\{1,\ldots,r\}$ we define $C(S) = V_0 \cup \bigcup_{i \in S} V_i$.

\begin{lemma}\label{lem:sizeSatellite} The {\sc Satellite Problem} for \size{} can be solved in $O(qn\log n+m)$ time. \end{lemma}
\begin{proof}
  Notice that $d(C(S)) = d(V_0) - \sum_{i \in S} d(V_i)$. Hence, we
  can reformulate the {\sc Satellite Problem} with $\mu = \size{}$ as
  finding a subset $S$ of $\{1, \ldots, r\}$ such that $\sum_{i \in S}
  d(V_i) \geq d(V_0)-q$ and $\sum_{i \in S} |V_i| \leq p-|V_0|$. Thus,
  we can associate with every $i$ an item with value $d(V_i)$ and
  weight $|V_i|$. The objective is to find a set of items with total
  value at least $d(V_0)-q$ and total weight at most $p-|V_0|$. This
  problem is known as {\sc Knapsack} and can be solved in $O(rv\log
  w)$ time by a classical dynamic
  programming~\cite{Bellman62,cormen2001introduction} algorithm, where
  $r$ is the number of items, $v$ is the value we seek to attain and
  $w$ is the weight limit. It is easy to compute the values and
  weights of all items in time $O(m)$. Since the number $r$ of items
  is at most $n$, the value is bounded from above by $q$ and the
  weight by $n$, the statement of the lemma follows.
\end{proof}

The case that $\mu = \nonedge{}$ is slightly more complicated, however we can still solve it using a polynomial-time dynamic programming algorithm. 

\begin{lemma}\label{lem:nonedgeSatellite} The {\sc Satellite Problem} for \nonedge{} can be solved in $O(pn^2m)$ time. \end{lemma}
\begin{proof}
Consider the set C(S) for a subset $S$ of $\{1, \ldots, i-1\}$. We investigate what happens to $\nonedge(C(S))$ and $d(C(S))$ when $i$ is inserted into $S$. For \nonedge{} we have the following equation.
\begin{equation}\label{eqn:nonedge1}
\nonedge(C(S \cup \{i\})) = \nonedge(C(S)) + \nonedge(V_i) + |C(S)| \cdot |V_i| - d(V_i)
\end{equation}
Furthermore, $d(C(S \cup \{i\})) = d(C(S)) - d(V_i)$. Define $T[i,j,k,\ell]$ to be true if there is a subset $S$ of $\{1, \ldots i\}$ such that $|C(S)|=j$, $d(C(S))=k$ and $\nonedge(C(S))=\ell$. If such a set $S$ exists, then either $i \in S$ or $i \notin S$. Together with Equation~\ref{eqn:nonedge1} this yields the following recurrence for $T[i,j,k,\ell]$. 
\begin{equation}\label{eqn:nonedge2}
T[i,j,k,\ell] = T[i-1,j,k,\ell] \vee T[i-1,j-|V_i|,k + d(V_i),\ell - \nonedge(V_i) - (j-|V_i|) \cdot |V_i| + d(V_i) ]
\end{equation}

The size of the table $T$ is $O(pn^2m)$ since $1 \leq i \leq r \leq n$, $0 \leq j \leq n$, $0 \leq k \leq m$, and $0 \leq \ell \leq p$, as it makes no sense to add more sets to $C$ after the threshold $p$ of non-edges in $C$ has been exceeded. 
We initialize the table to \textsf{\small true} in $T[0,|V_0|,d(V_0),\nonedge(V_0)]$
 and \textsf{\small false} everywhere else. Then we compute the values of the table using Equation~\ref{eqn:nonedge2}, treating every time we go out of bounds as a \textsf{\small false} entry. The algorithm returns \textsf{\small true} if there is an entry of $T$ which is \textsf{\small true} for $i=r$, $k \leq q$ and $\ell \leq p$. The running time bound is immediate, while correctness follows from Equations~\ref{eqn:nonedge1} and~\ref{eqn:nonedge2}. 
\end{proof}

For the version of {\sc Satellite Problem} when $\mu = \nondeg$ we do not have a polynomial time algorithm. Instead, we give an algorithm with running time $(3e)^{q+o(q)}n^{O(1)}$ based on dynamic programming and the {\em color coding} technique of Alon et al.~\cite{AlonYZ95}. 
When using color coding, it is common to give a randomized algorithm first, and then derandomize it using appropriate hash functions. In our case, existing hash functions are sufficient to give a deterministic algorithm, and our deterministic algorithm is not conceptually more difficult than the randomized version. Therefore, 
we only present the deterministic version. For this we will need the following proposition.
\begin{proposition}[\cite{NSS95}]\label{prop:hashFam1} For every $n$, $k$ there is a family of functions ${\cal F}$ of size $O(e^k \cdot k^{O(\log k)} \cdot \log n)$ such that every function $f \in {\cal F}$ is a function from $\{1, \ldots, n\}$ to $\{1, \ldots, k\}$ and for every subset $S$ of $\{1, \ldots, n\}$ of size $k$ there is a function $f \in {\cal F}$ that is bijective when restricted to $S$. Furthermore, given $n$ and $k$, ${\cal F}$ can be computed in time $O(e^k \cdot k^{O(\log k)} \cdot \log n)$.
\end{proposition}

\begin{lemma}\label{lem:nondegSatellite} 
There is a $(3e)^{q+o(q)}n^{O(1)}$ time algorithm for \nondeg{}-{\sc Satellite Problem}.
\end{lemma}

\begin{proof}
In Lemma~\ref{lem:sizeSatellite}, the set $S$ described which sets $V_i$ went into $C$. For this lemma, it is more convenient to let $S$ describe the sets $V_i$ which are {\em not} in $C$. Define $U=\{1, \ldots, r\}$, the task is to find a subset $S$ of $U$ such that $d(C(U \setminus S)) \leq q$ and $\nondeg(C(U \setminus S)) \leq p$. We iterate over all possible values $c \geq |V_0|$ of $|C(S)|$, and for each value of $c$ we will only look for sets $S$ such that $|C(U \setminus S)|=c$. This gives us the following advantage: for every vertex $v \in V_i$ for $i \geq 1$ if we choose to put $V_i$ into $C$ then $v$ will have exactly $c-d(v)-1$ non-neighbors in $C$. Hence for any $i$ such that $V_i$ contains a vertex $v$ with degree less than $c-p-1$ we know that $i \in S$. In other words, such a component $V_i$ should not be in the solution $C$, hence we can remove $V_i$ from the graph and decrease $q$ by $d(V_i)$ (as the edges $\Delta(V_i)\subseteq \Delta(V_0)$ will leave $C$ in any solution). Therefore, we can assume that every vertex $v\not\in V_0$ has degree at least $c-p-1$.

From now on, we only need to worry about $d(C(U \setminus S))\le q$ and about the non-degrees of vertices in $V_0$. A vertex $v \in V_0$ will have exactly $c-1-d(v)+|\Delta(v) \cap \Delta(C(U \setminus S))|$ non-neighbors. In particular, we need to make sure that no vertex $v \in V_0$ will have more than $p+d(v)-c+1$ neighbors outside of $C(U\setminus S)$. For every $v \in V_0$ we define the capacity of $v$ to be $\textup{cap}(v) = p+d(v)-c+1$. If any vertex has negative capacity, we discard the choice of $c$, as it is infeasible.


Every vertex $v \in V_0$ gets $\textup{cap}(v)$ {\em bins}. At this point we construct using Proposition~\ref{prop:hashFam1} a family ${\cal F}$ of colorings of the bins with colors from $\{1,\ldots,q\}$ such that for any set $X$ of $q$ bins there is a coloring $f \in {\cal F}$ that colors the bins in $X$ with different colors. The size of ${\cal F}$ is bounded from above by $O(e^q \cdot q^{O(\log q)} \cdot \log(n)) \leq O(e^{q+o(q)}\log(n))$. The algorithm has an outer loop in which it goes over all the colorings in ${\cal F}$. Every vertex $v$ in $V_0$ is assigned a set of colors, namely all the colors of the bins that belong to $v$. In the remainder of the proof we will assume that each vertex $v \in V_0$ has a set of colors attached to it. This set of colors is denoted by $\textup{colors}(v)$. Since $v$ had $\textup{cap}(v)$ bins assigned to it, we have that $|\textup{colors}(v)| \leq \textup{cap}(v)$.

In each iteration of the outer loop we will search for a special kind a solution: A map $\gamma$ that colors a set of edges in $\Delta(V_0)$ with colors from $\{1,\ldots,q\}$ is called {\em good} if the following two conditions are satisfied: (i) all edges that are colored by $\gamma$ receive different colors, and (ii) if an edge $e$ is colored by $\gamma$ and is incident to $v \in V_0$, then the color of $e$ is one of the colors of $v$. In other words, $\gamma(e) \in \textup{colors}(v)$. A subset $S \subseteq U$ is called {\em colorful} if the edges in $\Delta(C(U\setminus S))$ have a good coloring $\gamma$. What the algorithm will look for is a colorful set $S$ such that $|C(U \setminus S)|=c$ and a good coloring $\gamma$ of $\Delta(C(U\setminus S))$. Observe that since there are only $q$ different colors available and each edge of $\Delta(C(U\setminus S))$ must have a different color, a colorful solution automatically satisfies $d(C(U \setminus S)) \leq q$. Furthermore since every vertex $v \in V_0$ satisfies $|\textup{colors}(v)| \leq \textup{cap}(v)$, any colorful solution satisfies $\nondeg(C(U \setminus S)) \leq p$.

Conversely, consider a subset $S$ of $U$ such that $|C(U \setminus S)|=c$, $d(C(U \setminus S)) \leq q$ and $\nondeg(C(U \setminus S)) \leq p$. For each edge $e \in \Delta(C(U \setminus S))$, select a bin that belongs to the vertex $v \in V_0$ which is incident to $e$. Since each vertex $v \in V_0$ is incident to at most $\textup{cap}(v)$ edges in $\Delta(C(U\setminus S))$, we can select a different bin for each edge $e \in \Delta(C(U \setminus S))$. In total at most $q$ bins are selected, and hence there is an iteration of the outer loop where all of these bins are colored with different colors. Let $\gamma$ be a coloring of the edges in $\Delta(C(U \setminus S))$ that colors each edge with the color of the bin that the edge is assigned to. By construction, $\gamma$ is a good coloring of $\Delta(C(U \setminus S))$ in this iteration of the outer loop, and hence $S$ is colorful.



To complete the proof, we need an algorithm that decides whether there exists a colorful set $S \subseteq U$ such that $|C(U \setminus S)|=c$. 
For every $0 \leq i \leq r$, $0\leq j \leq n$ and $R \subseteq \{1, \ldots, q\}$, we define $T[i,j,R]$ to be \textsf{\small true} if there is a subset $S$ of $\{1, \ldots, i\}$ such that $|C(U \setminus S)|=j$, and a good coloring $\gamma$ of $\Delta(C(U\setminus S))$ with colors from $R$.
Suppose that such a set $S$ and map $\gamma$ exists. We have that either $i \in S$ or $i \notin S$. If $i \notin S$, then $S$ is a subset of $\{1, \ldots,i-1\}$ and hence $T[i-1,j,R]$ is \textsf{\small true}. If on the other hand $i \in S$, then let $S'=S \setminus \{i\}$ and $R_i$ be the set of colors of edges in $\Delta(V_i)$. In this case, we have that $|C(U \setminus S')|=j+|V_i|$, and $\gamma$ colors the edges of $\Delta(C(U\setminus S'))$ with colors from $R \setminus R_i$, so $T[i-1,j+|V_i|,R \setminus R_i]$ is true. Define ${\cal R}_i$ to be a family of sets of colors such that $R^* \in {\cal R}_i$ if there exists a good coloring of $\Delta(V_i)$ with colors from $R^*$. Clearly $R_i \in {\cal R}_i$. This yields the following recurrence for $T[i,j,R]$.
\begin{equation}\label{eqn:nondeg}
T[i,j,R] = 
T[i-1,j,R] \vee \bigvee_{\substack{R_i \in {\cal R}_i\\R_i \subseteq R}} T[i-1,j+|V_i|,R \setminus R_i] 
\end{equation}
Using Equation~\ref{eqn:nondeg} we can find a colorful $C$ in $3^qn^{O(1)}$ time as follows. We initialize the table to \textsf{\small true} in $T[0,|C(U)|,R]$ for all $R \subseteq \{1, \ldots, q\}$. Then we use Recurrence~\ref{eqn:nondeg} to fill the table for $T[i,j,R]$. The algorithm returns $\textsf{\small true}$ if $T[r,c,R]$ is true for some subset $R$ of $\{1,\ldots,q\}$. 
The running time of the algorithm for finding a colorful set $S$ is upper bounded by the size of the table, which is $2^qn^2$, times the time it takes to use Equation~\ref{eqn:nondeg} to fill a single table entry. To fill a table entry we go through all subsets $R_i \subset R$ and check whether $R_i \in {\cal R}_i$ in polynomial time by using a maximum matching algorithm. Specifically, we can build a bipartite graph with edges in $\Delta(V_i)$ on one side and elements of $R_i$ on the other. In this graph there is an edge between $e \in \Delta(V_i)$ and a color $r \in R_i$ if $e$ is incident to a vertex $v \in V_0$ such that $r \in \textup{colors}(v)$. Matchings in this graph that match all edges in $\Delta(V_i)$ to a color correspond exactly to good colorings of $\Delta(V_i)$ with colors from $R_i$. Thus the total running time is bounded by $O(\sum_{R \subseteq \{1,\ldots,q\}}\sum_{R' \subset R}n^{O(1)}) = O(3^qn^{O(1)})$. Correctness of the algorithm follows from Equation~\ref{eqn:nondeg}. The total runtime of the algorithm is bounded by $O(3^qn^{O(1)})$ times the number of iterations of the outer loop, which is $O(e^{q+o(q)}\log(n))$. This completes the proof of the lemma.
\end{proof}

Lemmata~\ref{lem:redstarderand2}, \ref{lem:sizeSatellite}, \ref{lem:nonedgeSatellite} and \ref{lem:nondegSatellite} give Theorem~\ref{lem:thmainq}.


\section{Parameterization by $p$}\label{sec:parp}
We prove in Section~\ref{sec:parpalg} that the \partitionsc{\mu}{p}{q}
is fixed-parameter tractable parameterized by $p$ for $\mu=\size$,
$\nonedge$, or $\nondeg$. Our algorithms work only on simple graphs,
i.e, graphs without parallel edges. In fact, as we show in
Section~\ref{sec:hardness-results}, the problem becomes hard if
parallel edges are allowed.
\subsection{Algorithms}\label{sec:parpalg}
In this section, we give algorithms with running time $2^{O(p)}n^{O(1)}$:
\begin{theorem}\label{thm:algp} There is a $2^{O(p)}n^{O(1)}$ time algorithm for \partitionsc{\size}{p}{q}, for \partitionsc{\nonedge}{p}{q} and for \partitionsc{\nondeg}{p}{q}.
\end{theorem}

 Because of Lemma~\ref{lem:uncrossing}, it is sufficient to solve the corresponding $(\mu,p,q)$-\textsc{Cluster} problem within the same time bound.  The setting is as follows. We are given a graph $G$, integers $p$ and $q$ and a vertex $v$ in $G$. The objective is to find a set $C$ {\em not} containing $v$ such that $d(C \cup \{v\}) \leq q$ and, depending on which problem we are solving, either $|C \cup \{v\}| = \size(C \cup \{v\}) \leq p$, $\nonedge(C \cup \{v\}) \leq p$ or $\nondeg(C \cup \{v\}) \leq p$.

For a set $S$ and vertex $v$, define $\Delta(S,v)$ to be the set of edges with one endpoint in $S$ and one in $\{v\}$. Define $\overline{\Delta}(S,v)$ to be $\Delta(S) \setminus {\Delta}(S,v)$, and let  $d(S,v)=|\Delta(S,v)|$ and $\overline{d}(S,v)=|\overline{\Delta}(S,v)|$. We will say that a set $C$ is $v$-{\em minimal} if $v \notin C$ and $d(C' \cup \{v\}) > d(C \cup \{v\})$ for every $C' \subset C$. As 
 \size{},\nonedge{} and \nondeg{} are monotone we can focus on $v$-minimal sets: if there is a solution for the cluster problem, then there is a solution of the form $C\cup\{v\}$ for some $v$-minimal set $C$.
The following fact uses that there are no parallel edges:
\begin{observation}\label{obs:cutsize} Let $C$ be a $v$-minimal set. Then $\overline{d}(C,v) < d(C,v) \leq |C|$. \end{observation}
In particular, if $\overline{d}(C,v) \geq d(C,v)$, then $d(v) \leq d(C \cup \{v\})$, contradicting that $C$ is $v$-minimal. Since $\overline{d}(C,v) < |C|$, it follows that $C$ must contain a vertex $u$ such that $N[u] \subseteq C \cup \{v\}$. Now we show that there are not too many $v$-minimal sets $C$ of size at most $p$ such that $G[C]$ is connected.

\begin{lemma}\label{lem:enumsmall} For any graph $G$, vertex $v$ and integer $p$, there are at most $4^pn$ $v$-minimal sets $C$ such that $|C| \leq p$ and $G[C]$ is connected. Furthermore, all such sets can be listed in time $O(4^pn)$.
\end{lemma}
\begin{proof}
By Observation~\ref{obs:cutsize}, any $v$-minimal set $C$ of size at
most $p$ satisfies $\overline{d}(C,v) < p$. Let $S$ be a set such that
$|S| \leq p$ and $G[S]$ is connected. Let $F$ be a subset of $N(S)
\setminus \{v\}$ of size at most $p-1$. We prove by downward induction
on $|S|$ and $|F|$ that there are at most $2^{2p-|S|-|F|-1}$
$v$-minimal sets such that $|C| \leq p$ , $G[C]$ is connected, $S \subseteq C$, and $F \cap C = \emptyset$. If $|S|=p$ then the only possibility for $C$ is $S$, while $2^{2p-|S|-|F|-1} \geq 1$. Similarly, consider the case that $|F|=p-1$. Now, every vertex of $F$ has at least one edge into $C$ and hence $F \cup \{v\} \subseteq N(C)$ and $\overline{d}(C,v) = p-1$. By Observation~\ref{obs:cutsize}, we have that $N(C)=F \cup \{v\}$ and the only possibility for $C$ is the connected component of $G \setminus (F \cup \{v\})$ that contains $S$. Hence there is at most one possibility for $C$ and $2^{2p-|S|-|F|-1} \geq 1$.

For the inductive step, consider a set $S$ such that $|S| \leq p$, and
$G[S]$ is connected and a subset $F$ of $N(S) \setminus \{v\}$ of size
at most $p-1$. We want to bound the number of $v$-minimal sets such
that $|C| \leq p$ and $G[C]$ is connected, $S \subseteq C$ and $F \cap
C = \emptyset$. If $N(S) \setminus (F \cup \{v\})$ is empty, then there is only one choice for $C$, namely $S$, and $2^{2p-|S|-|F|-1} \geq 1$. Otherwise, consider a vertex $u \in N(S) \setminus (F \cup \{v\})$. By the induction hypothesis, the number of $v$-minimal sets such that $|C| \leq p$ and $G[C]$ is connected, $S \cup \{u\} \subseteq C$ and $F \cap C = \emptyset$ is at most $2^{2p-|S|-|F|-2}$. Similarly, the number of $v$-minimal sets such that $|C| \leq p$ and $G[C]$ is connected, $S \subseteq C$ and $(F \cup \{u\}) \cap C = \emptyset$ is at most $2^{2p-|S|-|F|-2}$. Since either $u \in C$ or $u \notin C$, the two cases cover all possibilities for $C$ and hence there are at most $2 \cdot 2^{2p-|S|-|F|-2} = 2^{2p-|S|-|F|-1}$ possibilities for $C$.

For a fixed $S$ and $F$, the above proof can be translated into a procedure which lists all $v$-minimal sets such that $|C| \leq p$ and $G[C]$ is connected, $S \subseteq C$ and $F \cap C = \emptyset$. We run the procedure for $S=\{u\}$ and $F=\emptyset$ for every possible choice of $u$. Hence, there are at most $4^pn$ $v$-minimal sets $C$ such that $|C| \leq p$ and $G[C]$ is connected, and the sets can be efficiently listed. 
\end{proof}

The following observation is handy for using Lemma~\ref{lem:enumsmall}.

\begin{observation}\label{obs:unionminimal} Let $C$ be a $v$-minimal set of $G$ and $G[S]$ be a connected component of $G[C]$. Then $S$ is a $v$-minimal set. \end{observation}

In particular, if $S$ is not a $v$-minimal set, then it contains a
$v$-minimal set $S' \subset S$ and it is easy to see that $d(\{v\}
\cup (C \setminus S) \cup S') \leq d(\{v\} \cup C)$, contradicting the
minimality of $C$. Observation~\ref{obs:unionminimal} tells us that
any $v$-minimal set is the union of connected $v$-minimal sets. This
makes it possible to use Lemma~\ref{lem:enumsmall}.  We are now ready
to give an algorithm for \textsc{$(\size,p,q)$-Cluster}, the easiest
of the three clustering problems. Our algorithm is based on a
combination of color coding~\cite{AlonYZ95} with dynamic
programming.


\begin{lemma}\label{lem:algpsize}\textsc{$(\size,p,q)$-Cluster} can be
  solved in time $2^{O(p)}n^{O(1)}$.
\end{lemma}
\begin{proof}
We are given as input a graph $G$ together with a vertex $v$ and integers $p$ and $q$. The task is to find a vertex set $C$ of size at most $p-1$ such that $d(\{v\} \cup C) \leq q$. It is sufficient to search for a $v$-minimal set $C$ satisfying these properties. The algorithm of Lemma~\ref{lem:enumsmall} can be used to list all {\em connected} $v$-minimal sets $S_1, \ldots, S_\ell$ of size at most $p-1$; we have $\ell \leq 4^pn$. For a subset $Z$ of $\{1, \ldots, \ell\}$, define $C(Z) = \{v\} \cup \bigcup_{i \in Z} S_i$. 
By Observation~\ref{obs:unionminimal}, for any $v$-minimal set $C$ of size at most $p-1$ there exists a $Z \subseteq \{1, \ldots, \ell\}$ such that $C(Z)=C \cup \{v\}$. This set $Z$ satisfies the following properties. 
\begin{enumerate}\setlength\itemsep{-.7mm}
\item For every $i ,j \in Z$  with $i \neq j$, we have $S_i \cap S_j = \emptyset$.
\item $|C(Z)| = 1 + \sum_{i \in Z} |S_i| \leq p$.
\item $d(C(Z)) = d(v) + \sum_{i \in Z} (\overline{d}(S_i,v) - {d}(S_i,v)) \leq q$.
\end{enumerate}
However, for any subset $Z$ of $\{1, \ldots, \ell\}$ we have $|C(Z)| \leq 1 + \sum_{i \in Z} |S_i|$, and for any $Z$ that satisfies property (1), we have 
\[
d\big{(}C(Z) \big{)} \leq d(v) + \sum_{i \in Z} (\overline{d}(S_i,v) - {d}(S_i,v)).
\]
Hence it is sufficient to search for a set $Z$ such that for every $i ,j \in Z$  with $i \neq j$, we have $S_i \cap S_j = \emptyset$, $1 + \sum_{i \in Z} |S_i| \leq p$ and $d(v) + \sum_{i \in Z} (\overline{d}(S_i,v) - {d}(S_i,v)) \leq q$.





To ensure that the algorithm picks $Z$ such that the sets $S_i$ and $S_j$ will be disjoint for every pair of distinct integers $i,j \in Z$ we will use color coding. In particular, we construct a family ${\cal F}$ of functions from $V(G) \setminus \{v\}$ to $\{1, \ldots, p-1\}$ as described in Proposition~\ref{prop:hashFam1}. The family ${\cal F}$ has size $O(e^p \cdot p^{O(\log p)} \cdot \log n)$.
For each function $f \in {\cal F}$ we will think of the function as a
coloring of $V(G) \setminus \{v\}$ with colors from $\{1, \ldots,
p-1\}$. We will only look for a $v$-minimal set $C$ whose vertices
have different colors. This will not only ensure that any two sets
$S_i$ and $S_j$ that we pick will be disjoint, it also automatically
ensures that $1 + \sum_{i \in Z} |S_i| \leq p$. If
the input instance was a \textsf{\small yes}-instance then a solution
set $C$ exists, and the construction of ${\cal F}$ ensures that there
will be a function $f \in {\cal F}$ which colors all vertices in $C$
with different colors.

  When considering a particular coloring $f$, we discard all sets from
  $S_1, \ldots, S_\ell$ which have two vertices of the same color, so from
  this point, without loss of generality, all sets in $S_1, \ldots,
  S_\ell$ have at most one vertex of each color. For a vertex set $S$,
  define $\textsf{\small colors}(S)$ to be the set of colors occuring
  on vertices on $G$. For every $0 \leq i \leq \ell$, $0 \leq j \leq
  m$ and $R \subseteq \{1, \ldots, p-1\}$, we define $T[i,j,S]$ to
  be \textsf{\small true} if there is a set $Z \subseteq \{1, \ldots,
  i\}$ such that all vertices of $C(Z)$ have distinct colors, $d(v) +
  \sum_{i \in Z} (\overline{d}(S_i,v) - {d}(S_i,v)) = j$ and
  $\textsf{\small colors}(C(Z)) \subseteq R$. Clearly, there is a
  $v$-minimal set $C$ such that $d(\{v\} \cup C) \leq q$ and all
  vertices of $C$ have different color if and only if
  $T[\ell,j,\{1,\ldots,p-1\}]$ is true for some $j \leq q$. We can fill
  the table $T$ using the following recurrence.
\begin{equation*}
T[i,j,R] =
\begin{cases}
T[i-1,j,R] & \text{if } \textsf{\small colors}(S_i) \setminus R \neq \emptyset    \\
T[i-1,j,R] \vee T[i-1,&\\j+d(S_i,v) 
- \overline{d}(S_i,v), R \setminus \textsf{\small colors}(S_i)] & \text{otherwise } 
\end{cases}
\end{equation*}
Here we initialize $T[0,d(v),\emptyset]$ to true. The table has size $4^{p}n^{O(1)} \cdot 2^{p}n^{O(1)} = 8^{p}n^{O(1)}$ and can be filled in time proportional to its size. Hence the total running time for the algorithm is $(8e)^{p+o(p)}n^{O(1)}$.
\end{proof}

For \textsc{$(\size,p,q)$-Cluster} the size of the set $C$ we look for is already bounded by $p$.  For \textsc{$(\nonedge,p,q)$-Cluster} and \textsc{$(\nondeg,p,q)$-Cluster}, we cannot make this assumption, however the next lemma gives us a way to handle all large $v$-minimal sets $C$ for \textsc{$(\nonedge,p,q)$-Cluster} and \textsc{$(\nondeg,p,q)$-Cluster}.

\begin{lemma}\label{lem:enumlarge} For any graph $G$, vertex $v$ and integer $p$, there are at most $O(2^pn)$ $v$-minimal sets $C$ such that $\nondeg(C \cup \{v\}) \leq p$ and $|C| \geq 3p$. These sets can be listed in time $2^pn^{O(1)}$.
\end{lemma}
\begin{proof}
  By Observation~\ref{obs:cutsize}, any $v$-minimal set $C$ contains a
  vertex $u$ such that $N[u] \setminus \{v\} \subseteq C$. Thus, we go
  over every possibility for $u$, and we will enumerate all
  $v$-minimal sets $C$ of size at least $3p$ such that that $\nondeg(C
  \cup \{v\}) \leq p$ and $N[u] \setminus \{v\} \subseteq C$. Notice
  that $|C \setminus N[u]| \leq p$ since every vertex $w \in C
  \setminus N[u]$ is a non-neighbor of $u$. Thus, if $|C| \geq 3p$
  then $|N[u] \setminus \{v\}| \geq 2p$.  Hence, every vertex in $C
  \setminus N[u]$ has at least $|N[u] \setminus \{v\}|-p \geq p$ edges
  to $N[u] \setminus \{v\}$. Let $S$ be the set of vertices in $V(G)
  \setminus (N[u] \cup \{v\})$ which have at most $p$ non-neighbors
  in $N[u] \setminus \{v\}$.  If $|S| \geq p+3$, then no $v$-minimal set
  $C$ satisfying the constraints and the requirement $N[u] \setminus
  \{v\} \subseteq C$ can exist. To see this, suppose for contradiction
  that such a set $C$ exists, then $|S \setminus C| \geq 3$. Since each
  vertex in $S \setminus C$ has at most $p$ non-neighbors in $N[u]
  \setminus \{v\}$ it follows that
$$\overline{d}(C,v) \geq 3(|N[u] \setminus \{v\}|-p) \geq 3(|C|-2p) \geq |C|$$
contradicting Observation~\ref{obs:cutsize}. The second inequality
holds since $|N[u] \setminus \{v\}| \geq |C|-p$ and the third since
$|C| \geq 3p$. Finally, since $|S| \leq p+3$ and $C \setminus (N[u]
\setminus \{v\}) \subseteq S$ there are at most $2^{p+3}$ possible
sets $C$ for every choice of $u$. Thus there are at most
$O(2^pn)$ such sets, and they can be enumerated in
time $2^pn^{O(1)}$.
\end{proof}

Since every set $C$ such that $\nonedge(C) \leq p$ satisfies $\nondeg(C) \leq p$, we can use Lemma~\ref{lem:enumlarge} to find out whether there is a set $C$ such that $v \notin C$, $d(C \cup \{v\}) \leq q$, $\nonedge(C \cup \{v\}) \leq p$ and $|C| \geq 3p$ in time $2^pn^{O(1)}$. Thus we can concentrate on sets $C$ of size at most $3p$ for \textsc{$(\nonedge,p,q)$-Cluster} and for \textsc{$(\nondeg,p,q)$-Cluster}. For finding appropriate sets $C$ of size at most $3p$ for the two problems, we can give algorithms that are almost identical to the algorithm described in Lemma~\ref{lem:algpsize}. Since the two algorithms are so similar, we describe both in one go. 

\begin{lemma}\label{lem:algpnonedgedeg} There is a $2^{O(p)}n^{O(1)}$ time algorithm for the $(\nonedge,p,q)$-\textsc{Cluster} and the $(\nondeg,p,q)$-\textsc{Cluster} problems.
\end{lemma}
\begin{proof}
We are given as input a graph $G$ together with a vertex $v$ and integers $p$ and $q$. The task is to find a vertex set $C$ such that $d(\{v\} \cup C) \leq q$ and $\nonedge(\{v\} \cup C) \leq p$, or $\nondeg(\{v\} \cup C) \leq p$. It is sufficient to search for a $v$-minimal set $C$ satisfying these properties. Using Lemma~\ref{lem:enumlarge} we can check whether such a set of size at least $3p$ exists in time $2^pn^{O(1)}$. From now on, we only need to consider sets $C$ of size at most $3p$.

By Observation~\ref{obs:unionminimal}, $C$ can be decomposed into $C=S_1 \cup S_2 \ldots \cup S_t$ such that $S_i$ is a connected $v$-minimal set for every $i$, $S_i \cap S_j = \emptyset$ for every $i \neq j$, and no edge of $G$ has one endpoint in $S_i$ and the other in $S_j$ for every $i \neq j$. Using Lemma~\ref{lem:enumsmall} we list all connected $v$-minimal sets $S_1, \ldots, S_\ell$ of size at most $3p$ where $\ell \leq 4^{3p}n$. For a set $S$ and vertex $v$, define $\nondeg_v(S)$ to be the maximum number of non-edges to vertices in $S$ over all vertices in $S \setminus v$. For a subset $Z$ of $\{1, \ldots, \ell\}$ define $C(Z) = \{v\} \cup \bigcup_{i \in Z} S_i$. Now, let $Z \subseteq \{1, \ldots, \ell\}$ such that for every $i,j \in Z$ with $i \neq j$, we have $S_i \cap S_j = \emptyset$. We have that $|C(Z)| = 1 + \sum_{i \in Z} |S_i|$ and that
\begin{align}
\notag d(C(Z)) \leq d(v) + \sum_{i \in Z}(\overline{d}(S_i,v) - {d}(S_i,v)) \\
\label{eqn:nonedgeineq}\nonedge\big{(}C(Z) \big{)} \leq |C(Z)| + \sum_{i \in Z} \big{(}\nonedge(S_i) - d(S_i,v)\big{)} + \sum_{i \in Z} \sum_{i < j \in Z} |S_i||S_j| \\
\notag\nondeg\big{(}C(Z) \big{)} \leq \max\Big{\{} \sum_{i\in Z}\Big{(} |S_i|-d(S_i,v)\Big{)}, \max_{i \in Z}\Big{(} \nondeg_v(S_i \cup \{v\}) + \sum_{j \in Z \setminus \{i\}} |S_j|\Big{)}\Big{\}}
\end{align}

If there is no edge with one endpoint in $S_i$ and the other in $S_j$ for any $i \neq j$, $i \in Z$, $j \in Z$, then the inequalities hold with equality. Our algorithm will select a $Z$ such that $C = \bigcup_{i \in Z} S_i$. To ensure that the algorithm picks $Z$ such that the sets $S_i$ and $S_j$ will be disjoint for every pair of distinct integers $i,j \in Z$ we will use color coding. In particular, we construct a family ${\cal F}$ of functions from $V(G) \setminus \{v\}$ to $\{1, \ldots, 3p\}$ as described in Proposition~\ref{prop:hashFam1}. The family ${\cal F}$ has size $O(e^{3p} \cdot p^{O(\log p)} \cdot \log n)$.

For each function $f \in {\cal F}$ we will think of the function as a coloring of $V(G) \setminus \{v\}$ with colors from $\{1, \ldots, 3p\}$. We will only look for a $v$-minimal set $C$ whose vertices have different colors. This will ensure that any two sets $S_i$ and $S_j$ that we pick will be disjoint, and controls the total size of the set picked. If the input instance was a \textsf{\small yes}-instance, then a solution set $C$ exists, and the construction of ${\cal F}$ ensures that there will be a function $f \in {\cal F}$ which colors all vertices in $C$ with different colors. When considering a particular coloring $f$, we discard all sets from $S_1, \ldots S_\ell$ which have two vertices of the same color, so from this point, without loss of generality, all sets in $S_1, \ldots S_\ell$ have at most one vertex of each color. For a vertex set $S$, define $\textsf{\small colors}(S)$ to be the set of colors occuring on vertices on $G$. 

To solve \textsc{$(\nonedge,p,q)$-Cluster} we define a table $T_1$. For every $0 \leq i \leq \ell$, $0 \leq j \leq m$, $0 \leq k \leq p$ and $R \subseteq \{1, \ldots, 3p\}$ we define $T_1[i,j,k,S]$ to be \textsf{\small true} if there is a set $Z \subseteq \{1, \ldots, i\}$ such that all vertices of $C(Z)$ have distinct colors and
\begin{gather*}
d(v) + \sum_{i \in Z} (d(S_i,v) - \overline{d}(S_i,v)) \leq j,\\
 |R| + \sum_{i \in Z}\big{(} \nonedge(S_i) - d(S_i,v)\big{)} + \sum_{i \in Z} \sum_{i < j \in Z} |S_i||S_j| \leq k,\\
\textsf{\small colors}(C(Z)) = R.
\end{gather*}
For
\textsc{$(\nondeg,p,q)$-Cluster} we define a table $T_2$ in a similar
manner. That is, for every $0 \leq i \leq \ell$, $0 \leq j \leq m$, $0 \leq k \leq p$, $0 \leq x \leq p$, and $R \subseteq \{1, \ldots, 3p\}$ we define $T_2[i,j,k,x,S]$ to be \textsf{\small true} if there is a set $Z \subseteq \{1, \ldots, i\}$ such that all vertices of $C(Z)$ have distinct colors and
\begin{gather*}
d(v) + \sum_{i \in Z} (d(S_i,v) - \overline{d}(S_i,v)) \leq j,\\
\max_{i \in Z}\Big{(} \nondeg_v(S_i \cup \{v\}) + \sum_{j \in Z \setminus \{i\}} |S_j|\Big{)} \leq k,\\
\sum_{i \in Z}\Big{(} |S_i|-d(S_i,v)\Big{)} \leq x,\\
\textsf{\small colors}(C(Z)) = R.
\end{gather*}

There is a $v$-minimal set $C$ such that $d(\{v\} \cup C) \leq q$, $\nonedge(\{v\} \cup C) \leq p$ and all vertices of $C$ have different color if and only if $T_1[\ell,j,k,R]$ is true for some $j \leq q$, $k \leq p$ and $R \subseteq \{1, \ldots, 3p\}$. This follows directly from the definition of $T_1$ and the fact that Equation~\ref{eqn:nonedgeineq} holds with equality when there is no edge with one endpoint in $S_i$ and the other in $S_j$ for some $i \neq j$, $i \in Z$, $j \in Z$. By an identical argument, there is a $v$-minimal set $C$ such that $d(\{v\} \cup C) \leq q$, $\nondeg(\{v\} \cup C) \leq p$, and all vertices of $C$ have different color if and only if $T_2[\ell,j,k,x,R]$ is true for some $j \leq q$, $k \leq p$, $x \leq p$ and $R \subseteq \{1, \ldots, 3p\}$. We can fill the tables $T_1$ and $T_2$ using the following recurrences.
\begin{equation*}
T_1[i,j,k,R] =
\begin{cases}
T_1[i-1,j,k,R] & \text{if } \textsf{\small colors}(S_i) \nsubseteq R\\
T_1[i-1,j,k,R] \vee T_1[i-1,j',k_1',R']  & \text{otherwise } 
\end{cases} 
\end{equation*}
\begin{equation*}
T_2[i,j,k,\ell,R] =
\begin{cases}
T_2[i-1,j,k,x,R] & \text{if } \textsf{\small colors}(S_i) \nsubseteq R \\
T_2[i-1,j,k,x,R] & \text{if } |R|-|S_i|+\nondeg_v(S_i \cup \{v\}) > k \\
T_2[i-1,j,k,x,R] \vee T_2[i-1,j',k_2',x',R']  & \text{otherwise } 
\end{cases}
\end{equation*}
where $j' = j-d(S_i,v) + \overline{d}(S_i,v)$, $k_1' = k - \nonedge(S_i) - |S_i| + d(S_i,v) - |C_i|(|R|-|C_i|)$, $k_2' = k-|S_i|$, and $x'=x-|S_i|+d(S_i,v)$ and $R' = R \setminus \textsf{\small colors}(S_i)$. 

The recurrences above are correct for the following reason. Let $Z$ be a subset of $Z \subseteq \{1, \ldots, i-1\}$ such that all vertices of $C(Z)$ have distinct colors, $d(v) + \sum_{i \in Z} d(S_i,v) - \overline{d}(S_i,v) = j$, $|R| + \sum_{i \in Z} \nonedge(S_i,v) - d(S_i,v) + \sum_{i \in Z} \sum_{i < j \in Z} |S_i||S_j| = k,$ and $\textsf{\small colors}(C(Z)) = R$. Since every vertex in $C(Z)$ has a different color, it follows that $|C(Z)|=|R|$. Inserting $i$ into $Z$ is only possible if $\textsf{\small colors}(C(Z)) = R$. In that case, when we insert $i$ into $Z$, $d(v) + \sum_{p \in Z} (d(S_p,v) - \overline{d}(S_p,v))$ increases by $d(S_i,v) - \overline{d}(S_i,v)$. The sum $|R| + \sum_{i \in Z} \nonedge(S_i,v) - d(S_i,v) + \sum_{i \in Z} \sum_{i < j \in Z} |S_i||S_j|$ increases by $\nonedge(S_i) + |S_i| - d(S_i,v) + |C_i|(|R|-|C_i|)$. The expression $\max_{p \in Z}\big{(} \nondeg_v(S_p \cup \{v\}) + \sum_{j \in Z \setminus \{p\}} |S_j|\big{)}$ increases by $|S_i|$ or to $\nondeg_v(S_i \cup \{v\})+\sum_{j \in Z} |S_j|$, whichever yields the largest result. The expression $\sum_{p \in Z}\big{(} |S_p|-d(S_p,v)\big{)}$ increases by $|S_i|-d(S_i,v)$. Finally, the set of colors used, will now be $R \cup \textsf{\small colors}(S_i)$.

We initialize the tables $T_1$ and $T_2$ as follows. $T_1[0,j,k,\emptyset]$ is set to \textsf{\small true} for every $j \geq d(v)$, $k \geq 0$. 
$T_2[0,j,k,x,\emptyset]$ is set to \textsf{\small true} for every $j \geq d(v)$, $k \geq 0$, $x \geq 0$. 
Then we fill the tables using the recurrences above. The tables have size $4^{3p}n^{O(1)} \cdot 2^{3p}n^{O(1)} = 8^{3p}n^{O(1)}$ and can be filled in time proportional to their size. Hence the total running time for the algorithms is $(8e)^{3p+o(p)}n^{O(1)}$.
\end{proof}

Lemma~\ref{lem:uncrossing} together with Lemma~\ref{lem:algpsize} and Lemma~\ref{lem:algpnonedgedeg} yield Theorem~\ref{thm:algp}.

\subsection{Hardness results}\label{sec:hardness-results}
The algorithmic results in Section~\ref{sec:parq} still hold when parallel edges are allowed. Interestingly, the results in Section~\ref{sec:parpalg} do not: in particular, Observation~\ref{obs:cutsize} breaks down if there are parallel edges.  The following hardness result shows that allowing parallel edges indeed make the problems more difficult:

\begin{theorem}\label{thm:hardness} \partitionsc{\nonedge}{p}{q} and
  \partitionsc{\nondeg}{p}{q} are \textup{NP}-complete for $p=0$ on graphs with
  parallel edges. The problem \partitionsc{\size}{p}{q} is \textup{W[1]}-hard parameterized
  by $p$ on graphs with parallel edges (but in \textup{P} for every fixed $p$). \end{theorem}

To prove the Theorem~\ref{thm:hardness}, we reduce from the $k$-{\sc
  Clique} problem in $d$-regular graphs. Here we are given a graph $G$
in which all vertices have degree $d$, and an integer $k$. The task is
to find a clique $C$ in $G$ of size at least $k$. This problem is
NP-complete and W[1]-complete when parameterized by
$k$~(cf.~\cite{MathiesonS08}). Given a $d$-regular graph $G$ with $n$
vertices and $m=dn/2$ edges and an integer $k$, we construct a multigraph $G'$ by adding a vertex $v$ incident to all vertices of $G$ with
$m+1$ parallel edges to each vertex of $G$.

\begin{lemma} \label{lem:hardness2} There is a clique $C$ in $G$ of size $k$ if and only if there is a clique $C'$ in $G'$ such that $v \in C$ and $d(C) \leq (n-k)(m+1) + k(d-k)$. \end{lemma}
\begin{proof}
  In the forward direction, suppose $G$ contains a clique $C$ of size
  $k$. Then $C'=C \cup \{v\}$ is a clique in $G'$ and the number of
  edges leaving $C'$ is exactly $(n-k)(m+1) + k(d-k)$ (the first term
  is the number of edges leaving from $v$, the second term is the
  number of edges leaving from the $k$ vertices of $C$).

  In the backward direction, suppose there is a $C'$ in $G'$ such that
  $v \in C'$ and $d(C') \leq (n-k)(m+1) + k(d-k)$. Then $C = C'
  \setminus \{v\}$ is a clique in $G$. It suffices to argue that $|C|
  \geq k$. Suppose not. Then the number of edges leaving $C'$ is at
  least $(n-k+1)(m+1) > (n-k)(m+1) + k(d-k)$, a contradiction.
\end{proof}

Notice that $G'$ can be partitioned into cliques with at most $q =
(n-k)(m+1) + k(d-k)$ edges leaving each clique if and only if there is
a clique $C'$ in $G'$ such that $v \in C'$ and $d(C') \leq q$. If such
a clique $C'$ exists then there exists such a clique of size
$k+1$. Hence, by Lemma~\ref{lem:hardness2},
if \partitionsc{\nonedge}{0}{q} or \partitionsc{\nondeg}{0}{q} can be
solved in polynomial time on graphs with parallel edges, then {\sc
  Clique} in $d$-regular graphs can. Similarly,
if \partitionsc{\size}{p}{q} can be solved in $f(p)n^c$ time on graphs
with parallel edges then {\sc Clique} in $d$-regular graphs can be
solved in $f(k+1)n^c$ time. This proves Theorem~\ref{thm:hardness}.

Observe that Lemma~\ref{lem:uncrossing} automatically yields an
$O(n^{p+O(1)}m)$ time algorithm for \partitionsc{\size}{p}{q} on
graphs with parallel edges, and hence the lower bounds of
Theorem~\ref{thm:hardness} are tight.

\bibliographystyle{abbrv} 
\bibliography{cluster}







\end{document}